\documentclass[sigconf]{acmart}

\usepackage{algorithm}
\usepackage{booktabs}
\usepackage{multirow}
\usepackage{algpseudocode}
\usepackage{amsmath}
\usepackage{bbm}
\usepackage{epsfig}
\usepackage{graphicx}
\usepackage{hyperref}
\usepackage{adjustbox}
\usepackage{subcaption}
\usepackage{url}
\usepackage{xspace}
\usepackage{tikz}
\usepackage[framemethod=TikZ]{mdframed}
\usepackage{balance}
\usetikzlibrary{bayesnet}
\usetikzlibrary{positioning}

\hypersetup{
    unicode=false,          
    pdftoolbar=true,        
    pdfmenubar=true,        
    pdffitwindow=false,     
    pdfstartview={FitH},    
    pdftitle={My title},    
    pdfauthor={Author},     
    pdfsubject={Subject},   
    pdfcreator={Creator},   
    pdfproducer={Producer}, 
    pdfnewwindow=true,      
    colorlinks=true,       
    linkcolor=black,          
    citecolor=black,        
    filecolor=black,      
    urlcolor=black           
}

\newcommand{\spara}[1]{\smallskip\noindent{\bf #1}}

\newcommand{\para}[1]{\noindent{\bf #1}}
\newcommand{\epara}[1]{\noindent\emph{#1}}
\newcommand{\separa}[1]{\smallskip\noindent\emph{#1}}

\newtheorem{definition}{Definition}
\newtheorem{proposition}{Proposition}

\newtheorem{problem}{Problem}

\renewcommand{\vec}[1]{\ensuremath{\mathbf{#1}}\xspace}
\newcommand{\mat}[1]{\ensuremath{\mathbf{#1}}\xspace}
\newcommand{\eye}{\mat{I}}

\newcommand{\bigO}{\ensuremath{\mathcal{O}}\xspace}

\newcommand{\Din}{\ensuremath{\mat{D}_{\mathrm{in}}}\xspace}
\newcommand{\Dout}{\ensuremath{\mat{D}_{\mathrm{out}}}\xspace}

\newcommand{\convexset}{\ensuremath{\mathcal{C}}\xspace}
\newcommand{\feasibleset}{\ensuremath{\mathcal{A}}\xspace}

\algnewcommand\algorithmicforeach{\textbf{for each}}
\algdef{S}[FOR]{ForEach}[1]{\algorithmicforeach\ #1\ \algorithmicdo}

\renewcommand{\algorithmicrequire}{\textbf{Input:}}
\renewcommand{\algorithmicensure}{\textbf{Output:}}

\newcommand{\rhocounterfactual}{\ensuremath{\rho_{\mathrm{eq}}}\xspace}
\newcommand{\rhooverall}{\ensuremath{\rho_{0}}\xspace}

\newcommand{\ourmethod}{LcGD\xspace}

\newcommand{\squishlist}{
 \begin{list}{$\bullet$}
  {  \setlength{\itemsep}{0pt}
     \setlength{\parsep}{3pt}
     \setlength{\topsep}{3pt}
     \setlength{\partopsep}{0pt}
     \setlength{\leftmargin}{2em}
     \setlength{\labelwidth}{1.5em}
     \setlength{\labelsep}{0.5em}
} }
\newcommand{\squishlisttight}{
 \begin{list}{$\bullet$}
  { \setlength{\itemsep}{0pt}
    \setlength{\parsep}{0pt}
    \setlength{\topsep}{0pt}
    \setlength{\partopsep}{0pt}
    \setlength{\leftmargin}{2em}
    \setlength{\labelwidth}{1.5em}
    \setlength{\labelsep}{0.5em}
} }

\newcommand{\squishdesc}{
 \begin{list}{}
  {  \setlength{\itemsep}{0pt}
     \setlength{\parsep}{3pt}
     \setlength{\topsep}{3pt}
     \setlength{\partopsep}{0pt}
     \setlength{\leftmargin}{1em}
     \setlength{\labelwidth}{1.5em}
     \setlength{\labelsep}{0.5em}
} }

\newcommand{\squishend}{
  \end{list}
}
\usepackage[many]{tcolorbox}
\usepackage{xcolor}

\graphicspath{{figures/}}

\copyrightyear{2023}
\acmYear{2023}
\setcopyright{acmlicensed}\acmConference[CIKM '23]{Proceedings of the 32nd ACM International Conference on Information and Knowledge Management}{October 21--25, 2023}{Birmingham, United Kingdom}
\acmBooktitle{Proceedings of the 32nd ACM International Conference on Information and Knowledge Management (CIKM '23), October 21--25, 2023, Birmingham, United Kingdom}
\acmPrice{15.00}
\acmDOI{10.1145/3583780.3615025}
\acmISBN{979-8-4007-0124-5/23/10}

\title[Rebalancing Social Feed to Minimize Polarization and Disagreement]{Rebalancing Social Feed \\to Minimize Polarization and Disagreement}

\author{Federico Cinus}
\affiliation{%
  \institution{Sapienza University, Rome, Italy}
  \institution{CENTAI, Turin, Italy}
  \city{}
  \state{}
  \postcode{}
  \country{}
}
\email{federico.cinus@centai.eu}

\author{Aristides Gionis}
\affiliation{%
  \institution{KTH Royal Institute of Technology, Stockholm, Sweden}
  \city{}
  \state{}
  \postcode{}
  \country{}
}
\email{argioni@kth.se}

\author{Francesco Bonchi}
\affiliation{%
  \institution{CENTAI, Turin, Italy}
  \institution{Eurecat, Barcelona, Spain}
  \city{}
  \state{}
  \postcode{}
  \country{}
}
\email{bonchi@centai.eu}

\renewcommand{\shortauthors}{Federico Cinus, Aristides Gionis, \& Francesco Bonchi}

\begin{document}
\begin{abstract}
Social media have great potential for enabling public discourse on important societal issues.
However, adverse effects, such as polarization and echo chambers,
greatly impact the benefits of social media and
call for algorithms that mitigate these effects.
In this paper, we propose a novel problem formulation aimed at slightly nudging users' social feeds in order to strike a balance between relevance and diversity, thus mitigating the emergence of polarization, without lowering the quality of the feed. Our approach is based on re-weighting
the relative importance of the accounts that a user follows, so as to calibrate the frequency with which the content produced by various accounts is shown to the user.

We analyze the convexity properties of the problem, demonstrating the non-matrix convexity of the objective function and the convexity of the feasible set. To efficiently address the problem, we develop a scalable algorithm based on projected gradient descent. We also prove that our problem statement is a proper generalization of the undirected-case problem so that our method can also be adopted for undirected social networks. As a baseline for comparison in the undirected case, we develop a semidefinite programming approach, which provides the optimal solution.
Through extensive experiments on synthetic and real-world datasets, we validate the effectiveness of our approach, which outperforms non-trivial baselines, underscoring its ability to foster healthier and more cohesive online communities.
\end{abstract}



\maketitle \sloppy

\section{Introduction}
\label{sec:intro}
The phenomenon of polarization in social-media platforms indicates the emergence of two (or more) distinct poles of opinions, often coupled with structural network segregation~\cite{tucker2018social}. In fact, users of these platforms, tend to \emph{``follow''} other users whose interests and opinions align with their own, both because of natural homophily and because of \emph{who-to-follow} recommender systems.
Moreover, %
the algorithmically-curated \emph{``feed''} exposed to each user
ends up being constituted largely by content aligned with their own opinion.
As a result, algorithms employed by social-media platforms may lead to the well-known \emph{``echo-chamber''} effect~\cite{quattrociocchi2016echo,cinus2022effect},
where individuals of similar mindsets reinforce their beliefs while remaining shielded from dissenting viewpoints, thus becoming more polarized~\cite{nikolov2015measuring,pariser2011filter}.

Addressing the phenomenon of polarization and its negative consequences is crucial for fostering healthy online discourse and mitigating the spread of misinformation and extremism. Various strategies have been proposed to mitigate polarization in social networks \cite{hartman2022interventions}, including \emph{who-to-follow} recommendations aimed at minimizing controversy \cite{GarimellaMGM18},
or allocating seed users to optimize various diversity measures~\cite{AslayMGG18,garimella2017balancing,TuAG20}.
However, directly targeting and minimizing polarization can be problematic, as it may infringe upon users' freedom of expression and foster opposite effects \cite{bail2018exposure}. Moreover, recommending content or users to follow with a different mindset, might reduce the relevance of the content for the users, potentially eroding their experience on the platform.

In this paper, we study the problem of \emph{re\-balancing} the feed of users in a social network
with the aim of minimizing the \emph{polarization} and \emph{disagreement} in the network.
Our objective is inspired by the work of \citet{musco2018minimizing},
but we introduce a novel extension that poses significant challenges,
while being more realistic. In more detail,  we consider the setting of a \emph{directed} social network, where a link $(i,j)$ represents a \emph{follower-followee} relation, and thus user~$j$ can influence user~$i$.
We assume that opinions in the network are formed according to the
Friedkin-Johnsen opinion-dynamics model \cite{friedkin1990social}.
Our goal is to re-weight the relative importance of the accounts that a user follows,
so as to calibrate the frequency with which the content produced by various accounts is shown to the user. Our approach, by re-balancing the weight of existing links, preserves the total engagement of each user in the social network.

We formally characterize our problem, showing that the objective function is
not matrix-convex, but the feasible set is convex.
We propose a fast and scalable algorithm based on projected gradient descent,
which overcomes the challenges of non-matrix convexity and scalability.
We derive a closed-form expression for the gradient with respect to the adjacency matrix, and we exploit its rank-one properties to achieve an efficient evaluation through the Biconjugate Gradient Solver. Furthermore, we outline an efficient projection step and show that our method has approximately linear time complexity in the number of edges of the network.

We also prove that our problem is a proper generalization of the undirected case so that our method can be adopted for undirected graphs.
As a baseline for comparison in the undirected case, we also develop a semi\-definite programming approach.

We conduct extensive experimentation using both real-world and synthetic datasets, featuring diverse opinion distributions, varying polarization levels, and distinct modularity in the network structure. The results we obtain highlight the superior performance of our method compared to non-trivial baselines, providing empirical evidence of their efficacy in minimizing polarization and disagreement in large real-world social networks.

\smallskip 
\noindent The key contributions of this paper are as follows:

\squishlist
\item We propose a novel problem formulation for re-weighting the links of a directed social network with the goal to minimize polarization and disagreement, under preserving the total engagement for each user (Section \ref{sec:problem}).

\item We rigorously analyze the convexity properties of the problem. Specifically, we prove that the objective function is not matrix-convex while demonstrating the convexity of the feasible set (Section \ref{sec:problem}). We also show that our problem is a proper generalization of the undirected case.

\item For the proposed problem, we develop a projected gradient-descent algorithm that offers scalability and practicality (Section~\ref{subsec:method}). In particular, our method has near-linear time complexity in the number of edges per iteration (Section~\ref{subsec:runtime}).

\item As a baseline for studying the performance of our method to undirected social networks, we devise a semi\-definite programming approach that allows us to obtain an optimal solution for the undirected case (Section \ref{sec:undirected}),

\item We validate the effectiveness of our approach through a series of extensive experiments conducted on both synthetic and real-world large-scale datasets. We evaluate the performance across a wide range of input instances, providing a comprehensive characterization of its behavior (Section \ref{sec:experiments}).
\squishend

\section{Related Work}
\label{sec:related}
The phenomena of bias, polarization, and echo chambers have been studied extensively in social-media research~\cite{vicario2019polarization, bessi2016users, kubin2021role}. 
A main challenge has been to develop mitigation strategies
\cite{matakos2020tell, AslayMGG18, interian2021polarization,TuAG20, garimella2017balancing,baumann2023optimal}. 
However, many of the existing methods focus on a static setting, 
without considering the dynamic nature of the underlying opinion-formation process.

The \emph{Friedkin--Johnsen} (FJ) \emph{opinion-dynamics model} \cite{friedkin1990social} 
is the most widely-used framework for studying how opinions form among users in a social network  
\cite{acemoglu2011opinion, bindel2015bad, xia2011opinion, zha2020opinion}. 
In the FJ model, individuals are represented by nodes in a network
and social ties are represented by edges. 
Individuals update their opinions through an iterative weighted-averaging process 
that considers their own opinion and the opinions of their neighbors. 
The model incorporates two key variables: \emph{susceptibility} and \emph{strength of social ties}. 
While some previous research has studied the effects of susceptibility \cite{abebe2021opinion, marumo2021projected,xu2022effects}, 
in this paper, we assume that individual susceptibility is 
inherent to an individual's personality and, thus, it remains constant. 
Instead, we focus on \emph{structural interventions} as a means of shaping opinion dynamics. 
Structural interventions involve adjusting the weights of social ties, 
which can be interpreted as ranking scores, or frequencies of posts in users' social feeds.

The problem of jointly minimizing polarization and disagreement in the FJ model 
has been studied for undirected graphs in both discrete and continuous cases \cite{zhu2021minimizing,musco2018minimizing}. 
\citet{zhu2021minimizing} propose the problem 
of identifying $k$ new links in an undirected graph to minimize polarization and disagreement. 
They devise a greedy method to optimize this non-submodular objective function, 
which provides a constant-factor approximation and runs in cubic time. 
However, the practical applicability of their approach may be restricted due to its implicit assumptions on the number of link additions. 

In the continuous setting, \citet{chen2018quantifying} propose a novel framework for quantifying and optimizing conflict measures in social networks, including polarization, and disagreement. They investigate optimization potential across all sets of prior opinions in average- and worst-case scenarios, addressing challenges in determining globally optimal adjacency matrices using gradient descent. However, their approach focuses on small-graph modifications (deletions, insertions) and does not scale to large networks.
Also in the continuous setting, \citet{musco2018minimizing} investigate two problems, 
both of which differ from the one we study in this paper: 
the first problem involves optimizing graph topologies, 
while the second problem aims to optimize initial user opinions.

In the first problem studied by \citet{musco2018minimizing}, which is more closely related to our work, 
the objective is to find a weighted undirected %
graph that optimizes the objective function, given a total edge-weight budget. 
However, a key distinction from our formulation is that \emph{they do not consider a specific input network}. Instead, they explore all possible undirected graph topologies. 
Consequently, their formulation does not directly translate into actionable solutions such as link recommendation or social-feed curation.

The second key difference with the work of \citet{musco2018minimizing} lies in their focus on undirected graphs. 
In contrast, our formulation extends to directed (\emph{follower-followee}) social networks, 
which are more representative of real-world scenarios. 
The directed formulation is both \emph{more general}
and \emph{more meaningful} as, in general, 
the influence between two users is not symmetric. 
Furthermore, the non-symmetry of directed graphs poses significant new challenges, 
stemming from the non-convexity of the objective.

In Section \ref{subsec:undirected_stmt}, we define an undirected version of our problem and demonstrate that the directed version is a proper generalization. This allows our method to be applicable to undirected social networks. 
Our undirected formulation differs from the problem of \citet{musco2018minimizing} as it considers a given network structure and preserves per-user engagement, 
instead of distributing a total weight on a discovered graph topology. 

In our experiments (Section \ref{sec:experiments}), we compare our method against the SDP (optimal) formulation for undirected instances, which we define in Section \ref{sec:undirected}. The SDP formulation can be viewed as a version of the method proposed by \citet{musco2018minimizing}, adapted to accommodate a given graph topology. Our results demonstrate that our method outperforms the SDP formulation, as it provides the flexibility to assign different weights to $(i,j)$ and $(j,i)$, while the SDP only produces a symmetrical adjacency matrix.

\section{Problem definition}
\label{sec:problem}
We are given a directed, weighted graph $G=(V,E)$, with $|V|=n$ nodes and $|E|=m$ edges, where each node $i\in V$ corresponds to a user, and
each directed edge $(i,j)\in E$ indicates
that $i$ \emph{``follows''} $j$ or, in other terms, that  $j$ can influence the opinion of $i$. The edge weight $a_{ij}$ quantifies the amount of
influence that user $j$ exerts on user $i$.
We assume that $a_{ij} > 0$ if $(i,j)\in E$ and $a_{ij} = 0$ if $(i,j)\not\in E$
and we represent all the weights as a matrix $\mat{A}$, i.e., $\mat{A}[i,j] = a_{ij}$.

We adopt the popular \emph{Friedkin--Johnsen} (FJ) \emph{opinion-dynamics model}~\cite{friedkin1990social}. In the FJ model, each individual has an \emph{innate opinion}, which
may differ from their \emph{expressed opinion} on social media,
due to various factors, such as social pressure or fear of judgment.
For each individual $i\in V$,
their innate opinion is denoted by $s_i$ and their expressed opinion by $z_i$.
The sets of innate and expressed opinions, for all individuals in the network,
are represented by vectors $\vec{s}\in\mathbb{R}^{{n}}$ and $\vec{z}\in\mathbb{R}^{{n}}$, respectively. Individuals update their expressed opinions,
based on the opinions of their neighbors and their own innate opinion.
Specifically, for each individual $i$, their expressed opinion $z_i$ at time $t+1$
is given by the average of the opinions of their neighbors at time $t$
and their own innate opinion, weighted by the strength of their influence.
If we denote by \Dout the diagonal matrix whose $i$-th diagonal entry
contains the weighted out-degree of individual~$i$, i.e.,
$\Dout[i,i] = \sum_{j\in V} \mat{A}[i,j]$,
and by $\vec{z}^{(t)}$ the vector of expressed opinions at time $t$,
the opinion-update rule can be written in matrix notation as
\begin{equation}
\vec{z}^{(t+1)} = (\Dout+\eye)^{-1}(\mat{A}\vec{z}^{(t)} + \vec{s}).
\label{eq:update_matrix}
\end{equation}

By iterating Equation~\eqref{eq:update_matrix} and using the convergence theorems for matrices 
\cite[Theorem 7.17 and Lemma 7.18]{burden2015numerical}
we can find the equilibrium of the system,
where the opinions of all individuals have converged to a steady state.
This equilibrium is given by
\begin{equation}
\vec{z}^* = (\eye + \mat{L})^{-1}\vec{s},
\label{eq:equilibrium}
\end{equation}
where $\mat{L}=\Dout-\mat{A}$ is the Laplacian matrix of the graph $G$.
Equation~\eqref{eq:equilibrium} shows that the equilibrium opinions depend only on the innate opinions
and the structure of the social network, as captured by the Laplacian matrix.

Next, we define the measures of \emph{polarization} and \emph{disagreement} in a social network
for a given set of opinions \vec{z}.

\begin{definition}[Polarization]
The \emph{polarization index} of a set of opinions \vec{z}, denoted by $P(\vec{z})$,
measures the extent to which the opinions \vec{z} deviate from the average opinion
$\mu_\vec{z} = \frac{1}{{n}} \sum_{i=1}^{{n}} z_i$, that is,
\begin{equation}
P(\vec{z}) = \sum_i (z_i - \mu_\vec{z})^2.
\label{eq:polarization}
\end{equation}
\end{definition}

\begin{definition}[Disagreement]
The \emph{disagreement index} of a set of opinions $\vec{z}$ on a directed graph $G=(V,E)$,
denoted by $D(\vec{z})$,
measures the extent to which the opinions \vec{z} differ in the edges of~$G$.
It is defined as the average disagreement over all directed edges,
weighted by the amount of influence for each edge, that is,
\begin{equation}
D(\vec{z}, \mat{A}) = \frac{1}{2} \sum_{(i,j)\in E} a_{ij}(z_i-z_j)^2.
\label{eq:disagreement}
\end{equation}
\end{definition}

Given a social network $G$ and influence weights \mat{A} along its edges,
our goal in this paper is to slightly adjust the edge weights \mat{A} so as to minimize
the level of polarization and disagreement in the network.
Such edge-weight readjustment can help promote more balanced and constructive discussions among individuals.
As we discussed before, earlier works have studied similar tasks
but have focused on undirected graphs~\cite{chen2018quantifying, musco2018minimizing}.
In this paper, we focus on directed graphs
and methods proposed in previous work are not applicable.
Additionally, we require to modify the adjacency matrix
so that no new edges are introduced and the total out-degree of each node is preserved.
To quantify this requirement,
we denote by $\convexset(\mat{A})$ the set of matrices \mat{X}
for which
$\mat{X}\vec{1} = \mat{A}\vec{1}$ and
$\mat{A}[i,j]=0$ implies $\mat{X}[i,j]=0$.
Here, we write $\vec{1}$ to denote the vector of all ones.
We formally define $\convexset(\mat{A})$ in the next~section.

At a high level, we consider the following optimization problem,
which will be further refined in the next section.

\begin{problem}
\label{problem:rebalance}
Given a directed graph $G=(V, E)$ with weighted adjacency matrix \mat{A},
find a set of new edge weights $\mat{A}^*$
that minimizes the sum of polarization and disagreement at equilibrium, that is,
\[
\mat{A}^* = \arg\min_{\mat{X}\in\convexset(\mat{A})} P(\vec{z}^*) + D(\vec{z}^*, \mat{X}),
\]
where $\vec{z}^*$ is the equilibrium opinion vector obtained by the FJ model with edge weights $\mat{A}^*$,
and $\convexset(\mat{A})$ is defined above.
\end{problem}

\subsection{Characterization}
In this section, we aim to further characterize the problem we study,
focusing on its properties that stem by considering directed graphs.

\spara{Assumptions.}
We start by discussing some basic assumptions for our problem,
which we make mostly for ease of exposition.
First, we assume that the distribution of opinions $\vec{s} \in \mathbb{R}^{{n}}$
is \emph{mean centered}.
This assumption aligns with the literature \cite{musco2018minimizing,chen2018quantifying},
and incurs no loss of generality, as opinions can always be translated to have zero mean.
Furthermore, we assume that the  adjacency matrix \mat{A} of the input directed graph $G=(V, E)$
is \emph{row stochastic}, i.e., $\mat{A}\vec{1}=\vec{1}$.
This assumption allows for a straightforward interpretation of the total amount of influence
that a node receives to sum to $1$.
Based on these assumptions, the FJ model equilibrium in Equation~\eqref{eq:equilibrium} becomes
\begin{equation}
\vec{z}^* = (2\eye  - \mat{A})^{-1}\vec{s},
\label{eq:equilibrium2}
\end{equation}
where we have used the fact that for a \emph{row-stochastic matrix} \mat{A}
its Laplacian can be written as $\mat{L}=\eye-\mat{A}$.

\begin{proposition}
For zero-mean opinions
the \textit{polarization} and \textit{disagreement} indices
can be written using the following quadratic forms
\[
    P(\vec{z})=\vec{z}^{\top}\vec{z},
\quad\text{and}\quad
    D(\vec{z}, \mat{A})= \frac{1}{2}\vec{z}^{\top} (\eye + \Din - 2\mat{A})\vec{z},
\]
where \Din is the diagonal matrix containing the in-degrees of the nodes in $G$.
\end{proposition}

\begin{proof}
Using proposition 3.2 in the work of~\citet{musco2018minimizing}, since $\vec{s}$ is mean-centered, the average expressed opinion is zero, and the polarization measure is simply the sum of the squared opinions.
On the other hand, simply noticing that Equation~\eqref{eq:disagreement} is the sum of average edge disagreements (i.e., the squared distance between expressed opinions), the quadratic form can be derived as follows
\begin{align*}
2D(\vec{z}, \mat{A}) %
    & = \sum_{i,j}a_{ij} z_i^2 + \sum_{i,j} a_{ij}z_j^2 - 2\sum_{ij} a_{ij}z_i z_j \\
    & = \sum_{i,j} (1+d_j) z_iz_j \delta_{ij} - 2\sum_{ij} a_{i,j}z_i z_j \\
    & = \vec{z}^{\top} (\eye + \Din - 2\mat{A})\vec{z},
\end{align*}
where $d_j$ is in-degree of node $j$, and $\delta_{ij}$ is the Kronecker delta.
\end{proof}

\spara{Edge constraints.}
In contrast to previous studies \cite{musco2018minimizing,chen2018quantifying},
we restrict the feasible set of solutions to adjacency matrices where
\emph{the set of edges is a subset of the edges in the input graph} and each
\emph{individual out-degree is preserved}.
By doing so, we aim to mitigate polarization and disagreement by
\emph{using only pre-existing links} and
\emph{preserving the total engagement of each user} in the social network.
Formally, given the adjacency matrix \mat{A} for a graph $G$,
we define the convex set of feasible solutions as follows.
\begin{equation*}
\label{eq:feasible-set}
\convexset(\mat{A}) =
\{ \mat{X} \in \mathbb{R}^{{n}\times {n}} \mid
\mat{X}\vec{1}=\mat{A}\vec{1} \text{ and } \mat{A}[i,j]=0 \text{ implies } \mat{X}[i,j]=0 \}.
\end{equation*}
\begin{proposition}
\label{proposition:feasible-set-convexity}
The set $\convexset(\mat{A})$ is convex.
\end{proposition}
\begin{proof}
Let $\mat{X}_1,\mat{X}_2 \in \convexset(\mat{A})$ and $\lambda \in [0,1]$.
We need to show that the convex combination $\lambda \mat{X}_1 + (1-\lambda)\mat{X}_2 \in \convexset(\mat{A})$.
First, note that
$(\lambda \mat{X}_1 + (1-\lambda)\mat{X}_2)\mathbf{1} = \lambda \mat{X}_1\mathbf{1} + (1-\lambda)\mat{X}_2\mathbf{1} = \mathbf{1}$,
so the matrix is row-stochastic.
Second,  since $\mat{X}_1,\mat{X}_2 \in \convexset(\mat{A})$,
$\mat{A}[i,j]=0$ for some $(i,j)$ implies $\mat{X}_1[i,j]=\mat{X}_2[i,j]=0$,
and thus $(\lambda \mat{X}_1 + (1-\lambda)\mat{X}_2)[i,j]=0$.
Therefore, $\convexset(\mat{A})$ is convex.
\end{proof}

\para{Objective function.}
We define our objective as the sum of polarization and disagreement indices at the equilibrium
of the FJ opinion-dynamics model,
with the adjacency matrix \mat{A} being the variable of the optimization problem.
The following proposition describes our objective function.

\begin{proposition}
The sum of the polarization and disagreement indices at the equilibrium of the FJ model,
with a directed row-stochastic adjacency matrix \mat{A} and innate opinions $\vec{s}$,
is given by
\begin{equation}
f(\mat{A}, \vec{s}) = \vec{s}^\top (2\eye-\mat{A})^{-\top} \vec{s}+ \vec{s}^\top (2 \eye-\mat{A})^{-\top}\frac{(\Din-\eye)}{2}(2 \eye-\mat{A})^{-1}\vec{s}.
\label{eq:objective}
\end{equation}
\end{proposition}

\begin{proof}
Starting from Definitions \ref{eq:polarization} and \ref{eq:disagreement},
and considering the equilibrium vector as $\vec{z}^*$, we can express our objective function as:
\begin{align*}
f(\mat{A},\vec{s}) = P(\vec{z}^*) + D(\vec{z}^*, \mat{A})
= (\vec{z}^*)^{\top}\vec{z}^* + \frac{1}{2} (\vec{z}^*)^{\top} (\eye + \Din - 2A) \vec{z}^*.
\end{align*}
Using the equilibrium condition in Equation~\eqref{eq:equilibrium2}, we get:
\begin{align*}
f(\mat{A},\vec{s})
&= \vec{s}^{\top}(2\eye - \mat{A})^{-\top}\left( \frac{2\eye + \Din - 2A}{2} \right)  (2\eye - \mat{A})^{-1}\vec{s} \\
&= \vec{s}^{\top}(2\eye - \mat{A})^{-\top}\vec{s} + \vec{s}^{\top}(2\eye - \mat{A})^{-\top}\frac{ (\Din - \eye)}{2} (2\eye - \mat{A})^{-1}\vec{s}.
\end{align*}
\end{proof}

\epara{Remark:}
Since $\mat{A}$ is row-stochastic, each entry satisfies $|\mat{A}[i, j]| \leq 1$,
which implies that $\eye - \mat{A}$ is diagonally dominant.
Consequently, $2\eye-\mat{A}$ is strictly diagonally dominant, and always invertible.
However, the objective function is not convex, as shown in the following proposition.

\begin{proposition}
The objective in Equation~\eqref{eq:objective} is not a matrix-convex function.
\end{proposition}
\begin{proof}
To prove this proposition, we provide a counterexample by using the definition of a matrix-convex function. A matrix-valued function $f$ is said to be matrix-convex if and only if it satisfies the following inequality for all $\lambda \in [0, 1]$ and matrices $\mat{A}_1$ and $\mat{A}_2$:
\[
    f(\lambda \mat{A}_1 + (1-\lambda) \mat{A}_2) \leq \lambda f(\mat{A}_1) + (1-\lambda) f(\mat{A}_2)
\]
Consider a vector of opinions $\vec{s}^{\top} = (1 \: 0 \: -1)$.
Let $\mat{A}_1$ and $\mat{A}_2$ be two adjacency matrices of two connected graphs defined as:
\[
 \mat{A}_1 =
  \begin{bmatrix}
    0&1&0\\
    2/3& 0 &1/3 \\
    1 & 0 & 0
    \end{bmatrix},
\: \: \text{and} \: \:
\mat{A}_2 =
\begin{bmatrix}
    0&1&0\\
    1/3& 0 &2/3 \\
    0 & 1 & 0
    \end{bmatrix}.
\]
Setting $\lambda=0.5$, we can compute the two terms in the inequality to obtain
$f((\mat{A}_1 + \mat{A}_2)/2)=0.87$ and
$(f(\mat{A}_1) + f(\mat{A}_2))/2=0.84$.
Thus, the inequality is violated, and the objective function in Equation~\eqref{eq:objective} is not matrix-convex.
\end{proof}

We can now present the formal problem statement.

\begin{problem}%
Given a directed graph $G=(V, E)$ with row-stochastic adjacency matrix $\mat{A}$
and a vector of mean-centered innate opinions $\vec{s} \in \mathbb{R}^n$,
the goal is to find $\mat{A}^*$ that minimizes the following objective function,
which represents the polarization and disagreement index in the network,
\[
\min_{\mat{X} \in \convexset(\mat{A})}
s^\top  (2\eye-\mat{X})^{-\top} s+ s^\top (2 \eye-\mat{X})^{-\top}\frac{(\Din-\eye)}{2}(2 \eye-\mat{X})^{-1}s,
\]
where $\convexset(\mat{A})$ is the set of all row-stochastic sub-matrices of the adjacency matrix of $G$,
\Din is the diagonal matrix of in-degrees of the nodes of $G$, and \eye is the identity matrix.
\end{problem}

\subsection{The case of undirected graphs}
\label{subsec:undirected_stmt}

Directed graphs are particularly relevant to model social networks,
where the direction of edges captures the influence or flow of information among individuals.
However, undirected social networks are also interesting to study,
where the presence of a link between two individuals does not imply any directional relationship.
In the case of undirected graphs, the assumption of \emph{row stochasticity}
of the adjacency matrix is substituted with \emph{double stochasticity},
which imposes that both the in-degrees and out-degrees are equal to one.
Hence, we obtain the original definition of disagreement
as in earlier works~\cite{chen2018quantifying,musco2018minimizing}
based on the Laplacian matrix $D(\vec{z}, \mat{A})=\vec{z}^{\top}\mat{L}\vec{z}$.
As a consequence, it is worth noting that our definition of the objective function
(sum of polarization and disagreement) for directed graphs is a proper generalization
of the one for undirected graphs proposed earlier.
In fact, when $\Din=\eye$
the objective function in Equation~\eqref{eq:objective}
coincides with the definition for undirected graphs,
\[
f(\mat{A}, \vec{s}) = \vec{s}^{\top} (\eye+\mat{L})^{-1} \vec{s},
\:\: \text{ where } \: \mat{L}=\eye-\mat{A},
\]
making our approach applicable to both types of social networks.
Nonetheless, the matrix convexity of the objective function
in the undirected case makes it particularly interesting to study.
More details can be found in Section \ref{sec:undirected}.

\section{Algorithm}
\label{sec:methods}
In this section, we present an efficient algorithm that uses projected gradient descent to optimize a non-convex-constrained objective function with a matrix-valued variable. We delve into the gradient of the objective function and discuss its essential properties. Subsequently, we outline efficient evaluation techniques for both the objective and the gradient. We then describe the projection step, which ensures that the solution adheres to the imposed constraints. Furthermore, we discuss the scalability features of our method and provide a comprehensive comparison with existing literature on undirected graphs.

\subsection{Projected gradient method}
\label{subsec:method}

\spara{Gradient computation.}
Gradient-based approaches are first-order methods that compute the gradient of the objective and
update the solution with a small increment in the opposite direction of the gradient.
In our case, the variable is represented by an ${n}\times {n}$ matrix.

\begin{proposition}
The gradient of the objective in Eq.~\eqref{eq:objective} is
\begin{align*}
\frac{\partial f}{\partial \mat{A}} =\: &
     (2 \eye-\mat{A})^{-\top} \vec{s} \vec{s}^\top  (2 \eye-\mat{A})^{-\top} \\
     & + 2(2 \eye-\mat{A})^{-\top} \frac{({\Din}-\eye)}{2} (2 \eye-\mat{A})^{-1} \vec{s} \vec{s}^\top  (2 \eye-\mat{A})^{-\top} \\
     & + \frac{1}{2} \vec{1} ((2 \eye-\mat{A})^{-1}\vec{s})\odot (2 \eye-\mat{A})^{-1}\vec{s}))^{\top}.
\end{align*}
\end{proposition}
\begin{proof}
First, note that for the gradient of matrix-valued functions, we have
\cite[Eq.\ (2.2) and (2.9)]{brookes2005matrix}:
\begin{align}
\label{eq:differential}
\partial_{\mat{X}}(\vec{a}^\top \mat{X}^{-1}\vec{b}) &
    = -\mat{X}^{-\top} \vec{a}\vec{b}^\top \mat{X}^{-\top} \text{, and}\\
\label{eq:differential2}
\partial_{\mat{X}} (\vec{a}^\top \mat{X}^{-\top}\vec{b}) &
    = -(\mat{X}^{-1} \vec{a}\vec{b}^\top \mat{X}^{-1})^\top =  -\mat{X}^{-\top} \vec{b}\vec{a}^\top \mat{X}^{-\top}.
\end{align}

To prove the proposition,
we focus on each term of the objective separately.

\spara{I.}
By applying Eq.~\eqref{eq:differential2},
the gradient of the first term in Eq.~\eqref{eq:objective} is:
\begin{equation}
        \partial_\mat{A} (\vec{s}^\top (2\eye-\mat{A})^{-\top} \vec{s} ) = (2\eye-\mat{A})^{-\top} \vec{s}\vec{s}^\top (2\eye-\mat{A})^{-\top}.
        \label{eq:gradient_term1}
\end{equation}

\spara{II.}
For the second term of the objective,
we apply the gradient of the product \cite[Eq.~(37)]{petersen2008matrix}:
\begin{align}
\label{eq:sum_grads}
& \partial_{\mat{A}}\left[\vec{s}^{\top}(2\eye - \mat{A})^{-\top}\frac{ (\Din - \eye)}{2} (2\eye - \mat{A})^{-1}\vec{s}\right] \\
& = \partial_{\mat{A}}\!\left(\vec{s}^{\top}\!(2\eye - \mat{A})^{-\top}\vec{x}_1\right) +
         \partial_{\mat{A}}\!\left(\vec{x}_1^{\top}\!(2\eye - \mat{A})^{-1}\vec{s}\right) +
         \partial_{\mat{A}}\!\left(\vec{x_2}^{\top} \!\frac{ (\Din - \eye)}{2} \vec{x_2}\right), \nonumber
\end{align}
where we consider the terms $\vec{x}_1=\frac{ (\Din - \eye)}{2} (2\eye - \mat{A})^{-1}\vec{s}$ and
$\vec{x}_2=(2\eye - \mat{A})^{-1}\vec{s}$ as constant with respect to $\mat{A}$.

\spara{II.A.}
By applying Eq.~\eqref{eq:differential2} and \eqref{eq:differential},
we can compute the first two derivatives in Eq.~\eqref{eq:sum_grads} as follows:
\begin{equation}
\label{eq:gradient_term2}
        2(2 \eye-\mat{A})^{-\top}\frac{({\Din}-\eye)}{2}(2 \eye-\mat{A})^{-1}\vec{s}\vec{s}^\top(2 \eye-\mat{A})^{-\top}.
\end{equation}

\spara{II.B.}
For the last term in Eq.~\eqref{eq:sum_grads},
to compute the gradient of $(\mat{D}_{\mathrm{in}}-\eye)/2$ while keeping $\vec{x}_2$ constant, we first note that $\eye$ does not depend on $\mat{A}$, resulting in a zero gradient. Furthermore, we express $\Din$ as $\mathrm{diag}(\mat{A}^\top \vec{1})$.
We derive the gradient with respect to each entry in $\mat{A}$:
\begin{align*}
    \frac{d}{da_{ij}}\frac{1}{2} \left(\vec{x}_2^\top\mathrm{diag}(\mat{A}^\top \vec{1}) \vec{x}_2 \right)
        &  =\frac{1}{2}\vec{x}_2^\top  \left(\frac{d}{da_{ij}}\mathrm{diag}(\mat{A}^\top \vec{1})  \right) \vec{x}_2 \\
        & = \frac{1}{2}\vec{x}_2^\top  \delta_{ji} \vec{x}_2 = (\vec{x}_2 \odot \vec{x}_2)_{jj} \,.
\end{align*}
In the above derivation, we consider the following facts:
(i) The derivative of the diagonal matrix yields a value different than zero only on the diagonal, where it corresponds to a Kronecker delta. Additionally, since the matrix $\mat{A}$ is transposed, the derivative is $1$ only when the indices $(j,i)$ of $\delta_{ji}$ correspond to the indices of the differential $da_{ij}$.
(ii) The term $\vec{x}_2^\top \delta_{ji} \vec{x}_2$ represents the element-wise multiplication between the $j$-th entry of the row vector $\vec{z}$ and the $i$-th entry of the same vector, but transposed, when $j=i$. This operation is equivalent to multiplying each entry of the vector by itself, resulting in squared entries. Thus, the $(i,j)$ entry of the final gradient is the $j$-th element of the vector obtained from squaring the entries of $\vec{x}_2$. Consequently, we have:
\begin{equation}
        \partial_{\mat{A}}\left(\vec{x_2}^{\top} \frac{ (\Din - \eye)}{2} \vec{x_2}\right) = \vec{1} (\vec{x}_2 \odot \vec{x}_2)^\top
\label{eq:gradient_term3}
\end{equation}
Summing the terms in Eq.~\eqref{eq:gradient_term1}, \eqref{eq:gradient_term2}, and \eqref{eq:gradient_term3}  yields the result.
\end{proof}

\epara{Remark 1}: As previously stated, the presence of $2\eye$ in the objective function ensures the existence of the inverse and the continuity of the function with respect to $\mat{A}$.

\separa{Remark 2}: For the entries of $\mat{A}$ that are equal to zero
(due to the constraint on the feasible set)
the corresponding differentials are~zero.

\spara{Fast evaluation of the gradient and the objective.}
Moving forward, we discuss how to evaluate fast the gradient and the objective function.
We introduce three auxiliary vectors, $\vec{z_1}$, $\vec{z_2}$, and $\vec{z_3}$,
which are derived from the following linear systems:
\begin{align*}
\vec{z_1}=(2\eye-\mat{A}^{\top})^{-1}\vec{s}, & \quad
\vec{z_2}=(2\eye-\mat{A})^{-1}\vec{s}, \:\text{ and} \\
\vec{z_3}=(2\eye-&\mat{A}^{\top})^{-1}({\Dout}-\eye)\vec{z_2}.
\end{align*}
The gradient can be expressed as the sum of three external products
\begin{equation}
     \frac{\partial f}{\partial \mat{A}} = \vec{z_1}\vec{z_2}^\top+\vec{z_3}\vec{z_2}^\top + \vec{1}(\vec{z_2}^2)^\top.
     \label{eq:gradient2}
\end{equation}
Additionally, the objective function can be written as the sum of two inner products
\begin{equation}
    f(\mat{A}, \vec{s})=\vec{s}^\top \vec{z_1} + \vec{s}^\top \vec{z_3}.
    \label{eq:objective2}
\end{equation}

To avoid the need for matrix inversion,
we propose solving each linear system of the form
$(2\eye-\mat{X})\vec{z}=\vec{s}$ directly using a linear solver.
Specifically, we employ the Biconjugate Gradient (BiCG)  solver \cite{fletcher1976numerical},
leveraging the sparsity of the matrix and the efficient instantiation of the solver,
at each iteration as proposed by \citet{marumo2021projected}.

\spara{Projection into the feasible set.}
The projection step plays a crucial role in finding the closest element
in the feasible set $\convexset(\mat{A})$,
which we shorthand by \feasibleset.
This step can be efficiently accomplished by sequentially applying the following three projections:
\begin{enumerate}
\item set all negative entries to zero, as the adjacency matrix should be non-negative;
\item set all edges not in $E$ to zero,
    ensuring that only existing edges remain in the matrix:
    for all $(i,j)\not\in E$ set $\mat{A}[i,j]=0$;
\item normalize each row of the adjacency matrix by dividing it by its $\ell_1$ norm.
\end{enumerate}

\spara{Algorithm.} The algorithm, presented in Algorithm \ref{alg:algorithm}, follows these steps:
(i) Initialize the solution with the input weights and an empty set of objective values,
denoted as $\Gamma$ (lines 1--2).
The input weights can be chosen uniformly at random, or using the time-inverted FJ model;
details are provided in Section \ref{sec:experiments}.
(ii) While the reduction in the objective value remains greater than a predefined threshold $\delta$, the algorithm performs the following steps:

\separa{\textbf{S1.}} Invoke the Biconjugate Gradient (BiCG) algorithm to solve the linear systems
$(2\eye+\mat{X})\vec{z}=\vec{s}$ for each iteration,
facilitating the computation of the gradient as specified in Eq.~\eqref{eq:gradient2} (lines 5--7).

\separa{\textbf{S2.}} Take a step of gradient descent and project the solution into the feasible set
using the algorithm defined in Algorithm \ref{alg:proj} (lines 9--10). Modify the solution proportionally to the input budget (line 11).

\separa{\textbf{S3.}} Store the objective value by utilizing the solutions obtained in lines 5 and 7, allowing for efficient evaluation of the objective function (line 11).

\begin{algorithm}[t]
   \footnotesize
    \caption{Laplacian-Constrained Gradient Descent (\ourmethod)}
    \begin{flushleft}
    \algorithmicrequire \; $\mat{A}^{\mathrm{init}}$; opinions $\vec{s}$, tolerance $\delta$, step-size $\eta$, budget $\beta$. \\
    \algorithmicensure \;$\mat{A}$.
    \begin{algorithmic}[1]
        \State  $ \mat{A}[i,j]\leftarrow \mat{A}^{\mathrm{init}}[i,j]$
        \State $\Gamma \leftarrow \emptyset$
        \State $t\leftarrow 1$
        \While{$\Gamma^{(t-1)} - \Gamma^{(t)} > \delta$}
            \State $\vec{z_1} \leftarrow \mathrm{BiCG\_SOLVER} (2\eye-\mat{A}^{\top}, \vec{s})$
\State $\vec{z_2}\leftarrow \mathrm{BiCG\_SOLVER}(2\eye-\mat{A},\vec{s})$
\State $\vec{z_3}\leftarrow \mathrm{BiCG\_SOLVER}(2\eye-\mat{A}^{\top}, ({\Dout}-\eye)\vec{z_2})$
            \State $ \partial_A f \leftarrow \vec{z_1}\vec{z_2}^{\top}+\vec{z_3}\vec{z_2}^{\top} + \vec{1}(\vec{z_2}^2)^{\top} $  \hfill{Gradient Eq.~\eqref{eq:gradient2}}
            \State $\Tilde{\mat{A}} \leftarrow \mat{A} - \eta \partial_A f$
            \State $\mat{A}_{\mathrm{proj}} \leftarrow \mathrm{proj}_{\mathrm{{\feasibleset}}}
            (\Tilde{\mat{A}})$ \hfill{Projection step}
            \State $\mat{A} \leftarrow \beta \mat{A}_{\mathrm{proj}} + (1-\beta)\mat{A}^{\mathrm{init}}$ \hfill{Partial modification}
            \State $  f(\mat{A}, \vec{s}) \leftarrow \vec{s}^\top \vec{z_1} + \vec{s}^\top \vec{z_3}$ \hfill{Current objective value Eq.~\eqref{eq:objective2}}
            \State $\Gamma \leftarrow \Gamma \cup \{ f(\mat{A}, \vec{s})\}$
            \State $t\leftarrow t + 1$

        \EndWhile \\
        \Return $\mat{A}$.
    \end{algorithmic}
    \end{flushleft}
  \normalsize
  \label{alg:algorithm}
\vspace{-1mm}
\end{algorithm}

\begin{algorithm}[t]
   \footnotesize
    \caption{Project on the feasible solution set ($\mathrm{proj}_{\mathrm{{\feasibleset}}}$)}
    \begin{flushleft}
    \algorithmicrequire \; $\Tilde{\mat{A}}$, $E$. \\
    \algorithmicensure \;$\mat{A}$.
    \begin{algorithmic}[1]
        \State  $\Tilde{\mat{A}}[i,j] \leftarrow 0  \:
            \text{ for all } i, j \text{ such that } \Tilde{\mat{A}}[i,j]<0$
            \hfill{Removing negative weights}
        \State  $\Tilde{\mat{A}}[i,j] \leftarrow 0  \: \text{ for all } i, j \not\in {E}$ \hfill{Removing non-existing edges}
        \State  $ \mat{A} \leftarrow \Tilde{\mat{A}} \oslash (\Tilde{\mat{A}} \vec{1}) $ \hfill{Normalizing rows}\\
        \Return $\mat{A}$
    \end{algorithmic}
    \end{flushleft}
  \normalsize
\vspace{-1mm}
\label{alg:proj}
\end{algorithm}

\subsection{Run-time complexity}\label{subsec:runtime}
Algorithm \ref{alg:algorithm} demonstrates near-linear time complexity in the number of edges per iteration, as supported by the following proposition:
\begin{proposition}
The run-time complexity for each step of Alg. \ref{alg:algorithm} is $\bigO((3T+2){n} + (3T+4){m})$,
where $T$ denotes the number of iterations of the
Biconjugate Gradient Stabilized (BiCGStab)  solver~\cite{marumo2021projected}.
\end{proposition}
\begin{proof}
We start with the projection algorithm presented in Algo\-rithm \ref{alg:proj}.
The variable $\mat{A}$ is stored in a sparse format,
requiring $\bigO({m})$ space, which leads to $\bigO({m})$ time for line 1.
In practice, this computation is significantly reduced due to the fact that decreasing polarization and disagreement necessitates increased communication among users,
resulting in fewer null edge weights.
Furthermore, line 2 is skipped in practice since non-existing edges are not stored.
Last, the row-normalization step in line 3 involves summing over the neighbors for each node,
resulting in a $\bigO({m})$ time complexity for the entire graph.

Moving to Algorithm \ref{alg:algorithm},
the initialization assignments in line 1 require $\bigO({m})$ time.
In each iteration, the Biconjugate Gradient Stabilized (BiCGStab) algorithm solves one of the three linear systems in time $\bigO(T({n}+{m}))$,
where $T$ denotes the number of iterations of the solver.
In line 8, only the ${m}$ entries of each of the three outer products are computed and stored,
as per the constraint in the feasible set.
The projection in line 10 requires $\bigO({m})$ time, as mentioned earlier,
and the inner products in line 11 take time $\bigO({n})$.
Therefore, the total time of each iteration is $\bigO((3T+2){n} + (3T+4){m})$.
\end{proof}

\subsection{The case of undirected graphs}
\label{sec:undirected}

So far we have presented our solution for minimizing polarization and disagreement for directed graphs.
While this is the first formulation that addresses the problem for directed graphs,
it is important to establish a connection with the existing literature,
so as to enable a meaningful comparison and evaluation of the results in the experimental section.
To ensure comparability between the different solutions,
we consider the doubly-stochastic version of the adjacency matrix.
By adopting this approach, we align our formulation with the literature,
which typically focuses on undirected graphs and utilizes symmetric adjacency matrices.

Consider an undirected graph $G=(V, E)$ with a symmetric doubly-stochastic adjacency matrix $\mat{A}$ and a vector of inner opinions $\vec{s} \in \mathbb{R}^n$.
The objective is to find a symmetric matrix $\mat{A}^*$ such that
\begin{equation}
\label{eq:problem}
\mat{A}^* = \arg \min_{\mat{X} \in \convexset(\mat{A})}
\vec{s}^{\top} (2\eye-\mat{A})^{-1}\vec{s},
\end{equation}
where the feasible set $\convexset(\mat{A})$ is now defined as
\begin{align*}
\label{eq:feasible-set2}
\convexset(\mat{A}) = \{ \mat{X} \in \mathbb{R}^{{n}\times {n}} \mid  \:\: &
    \mat{X}\vec{1}=\vec{1} \text{ and } \mat{X}^{\top}\vec{1}=\vec{1} \text{ and } \\
    &  \mat{A}[i,j]=0 \text{ implies } \mat{X}[i,j]=0  \}.
\end{align*}

It has been shown in the literature \cite{nordstrom2011convexity, haynsworth1970applications}
that the objective function of this problem is convex,
while the convexity of the feasible set can be easily shown,
mutatis mutandis to Proposition \ref{proposition:feasible-set-convexity}.
Therefore, an approximate solution can be obtained in polynomial time using semi\-definite programming (SDP).

\begin{figure*}[t!]
\centering
\includegraphics[width=.9\textwidth]{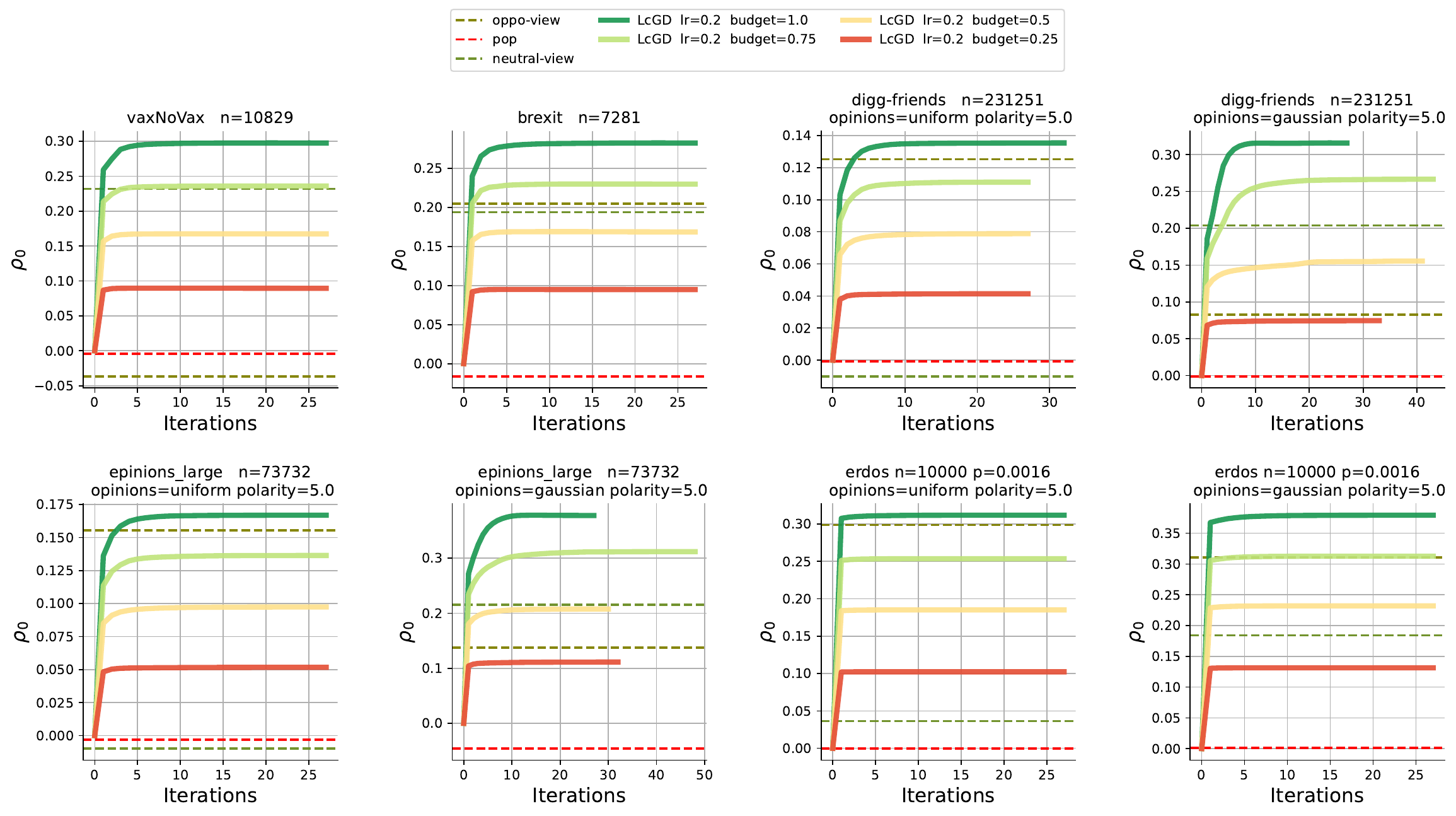}
\vspace{-4mm}
\caption{Reduction in polarization and disagreement (in the counterfactual scenario) vs number of iterations, using the stopping condition $\delta = 10^{-6} |E|$. For vaxNoVax and Brexit networks, we have real-world estimated opinions. For the other semi-synthetic datasets we report both uniform and gaussian distributions of opinions. }
\label{fig:objective_curves}
\vspace{-2mm}
\end{figure*}

\spara{Semidefinite programming.}
In SDP, we aim to minimize a linear function subject to constraints of the form
$\mat{A}\vec{x} = \vec{b}$ and $\mat{X} \succeq 0$,
where $\mat{X}$ is a symmetric matrix, and
the notation $\mat{X} \succeq 0$ indicates that $\mat{X}$ is positive semi\-definite.
In the case of undirected graphs, the matrix variable $\mat{A}$ is positive semi\-definite.
To demonstrate that the problem can be formulated as an SDP,
we refer to the classic textbook of \citet[Example 3.4, page 76]{boyd2004convex}.
The general problem we want to solve can be written as
\begin{equation}
\min_{t \in \mathrm{epi}\, f(\mat{A})} t ,
\label{eq:sdp}
\end{equation}
where $t\in \mathbb{R}$, and the \emph{epigraph} of the objective function
$f(\mat{A}) = \vec{s}^{\top} (2\eye-\mat{A})^{-1} \vec{s}$ is defined as
\begin{align*}
\mathrm{epi}\, f & =
\{ (\vec{s},2\eye-\mat{A},t) \mid \eye+\mat{L} \succ 0, \: \vec{s}^\top (2\eye-\mat{A})^{-1} \vec{s}\leq t\} \\
& = \{ (\vec{s},2\eye-\mat{A},t) \mid
    \begin{bmatrix}
        2\eye-\mat{A} & \vec{s} \\
        \vec{s}^{\top} & t
    \end{bmatrix} \succeq 0, \: 2\eye-\mat{A} \succ 0 \},
\end{align*}
where we utilize the Schur's complement.
The last condition represents a linear matrix inequality in $(\vec{s}, 2\eye-\mat{A}, t)$,
and the set is convex.
This convex-optimization problem can be efficiently solved,
approximately within any constant factor,
using various first-order methods such as SCS~\cite{o2016conic} and ADMM~\cite{wen2010alternating}.

To facilitate the solution process, several standard SDP libraries,
including CVX \cite{grant2014cvx},
implement these exact solvers.
It is worth noting, however, that these libraries may encounter scalability issues,
as highlighted in recent research \cite{majumdar2019survey}.

\section{Experimental evaluation}
\label{sec:experiments}
We structure our experimental evaluation to answer the following research questions:

\begin{itemize}
\item[\textbf{RQ1:}] What is the reduction in polarization and disagreement for different input networks?
How does this reduction compare to the equilibrium state with no interventions?
\item[\textbf{RQ2:}] What is the effect of \emph{network modularity} and \emph{initial polarization} on the behavior of the algorithm? What is the effect of the budget parameter?
\item[\textbf{RQ3:}] How does the algorithm perform for undirected graphs?
\end{itemize}

\begin{table}[t]
\caption{Network statistics.}
\vspace{-3mm}
\centering
\resizebox{\columnwidth}{!}{\begin{tabular}{lcrrrl}
\toprule
 Network name & Type & \# nodes & \# edges & Description \\
\midrule
brexit & R & 7\,281 & 530\,608 & retweets (UK Brexit referendum) \\
vaxNoVax & R &10\,829 & 1\,430\,776&  retweets (Italian vaccination debate)\\
Albert-Barab\'{a}si & S&  7\,047 & 10\,763 & synthetic scale-free graph\\
SBM & S & 9\,999 & 81\,867 &  synthetic stochastic block model\\
Erd\H{o}s-R\'{e}nyi & S& 10\,000 & 160\,107 &  synthetic binomial graph\\
ego-twitter & SS & 11\,548 & 17\,100 &  Twitter following network \\

facebook-wosn-wall & SS& 37\,082 & 260\,307 &   Facebook posts to other user's wall\\
soc-epinions1 & SS & 49\,604 & 471\,855 &   trust and distrust network \\
 epinions &  SS & 73\,732 & 767\,982 &  trust network from Epinions\\
munmun & SS & 464\,882 & 834\,548 &  Twitter mentions\\
 digg-friends & SS & 231\,251 & 1\,666\,463 & directed friendship graph of Digg\\
 flickr-growth & SS & 2\,070\,819 & 31\,335\,568 & Flickr social network \\
\bottomrule
\end{tabular}
}
\label{table:stats}
\vspace{-3mm}
\end{table}

\subsection{Experimental setup}
In this section, we describe our experimental pipeline.
We begin by describing the types of data we consider.
We then present our baselines and evaluation measures,
and finally, we present the overall experimental scheme.
Our code is made publicly available.\footnote{\url{https://github.com/FedericoCinus/rebalancing-social-feed}}

\begin{table*}[t]
\caption{Results on the reduction of polarization and disagreement for directed graphs with Gaussian opinion distribution, using a learning rate $\eta=0.2$ and stopping rule  $\delta=10^{-6}|E|$.
} \label{tab:datasets}
\vspace{-3mm}
\let\center\empty
\let\endcenter\relax
\centering
\resizebox{.75\width}{!}{\begin{tabular}{lrrrrrrrrrrrrr}
\toprule
\multicolumn{2}{c}{Input configuration} &  \multicolumn{4}{c}{\rhocounterfactual 
 ($\uparrow$ \textbf{better})} & \multicolumn{4}{c}{\rhooverall ($\uparrow$ \textbf{better})} & \multicolumn{4}{c}{Time (sec)} \\ \cmidrule(lr){1-2}\cmidrule(lr){3-6}\cmidrule(lr){7-10}\cmidrule(lr){11-14}
 \multicolumn{1}{l}{Network} & \multicolumn{1}{c}{Polarity} & \ourmethod & neutral-view & oppo-view & pop & \ourmethod & neutral-view & oppo-view & pop & \ourmethod & neutral-view & oppo-view & pop \\
 \midrule
\multirow[c]{2}{*}{Albert-Barab\'{a}si} & 1.0 & \textbf{0.2375} & 0.1688 & 0.0677 & -0.0710 & \textbf{0.8195} & 0.8037 & 0.7784 & 0.7493 & 3.15 & 0.00 & 0.00 & 0.00 \\
 & 5.0 & \textbf{0.4001} & 0.2215 & 0.2845 & -0.1007 & \textbf{0.8008} & 0.7415 & 0.7624 & 0.6346 & 7.04 & 0.00 & 0.00 & 0.00 \\
\multirow[c]{2}{*}{brexit} & 1.0 & \textbf{0.2824} & 0.1937 & 0.2045 & -0.0162 & \textbf{0.8216} & 0.7996 & 0.8023 & 0.7474 & 43.08 & 0.02 & 0.04 & 0.02 \\
 & 5.0 & \textbf{0.2824} & 0.1937 & 0.2045 & -0.0162 & \textbf{0.8216} & 0.7996 & 0.8023 & 0.7474 & 38.09 & 0.02 & 0.06 & 0.03 \\
\multirow[c]{2}{*}{digg-friends} & 1.0 & \textbf{0.1810} & 0.1164 & 0.0251 & -0.0097 & \textbf{0.8287} & 0.8151 & 0.7960 & 0.7887 & 639.54 & 0.08 & 0.04 & 0.04 \\
 & 5.0 & \textbf{0.3104} & 0.1931 & 0.0916 & 0.0039 & \textbf{0.7830} & 0.7461 & 0.7141 & 0.6865 & 354.45 & 0.07 & 0.07 & 0.06 \\
\multirow[c]{2}{*}{ego-twitter} & 1.0 & \textbf{0.0747} & 0.0193 & 0.0159 & -0.0114 & \textbf{0.8174} & 0.8065 & 0.8058 & 0.8004 & 60.03 & 0.00 & 0.00 & 0.00 \\
 & 5.0 & \textbf{0.1478} & 0.0401 & 0.0540 & -0.0137 & \textbf{0.6671} & 0.6250 & 0.6304 & 0.6040 & 79.31 & 0.00 & 0.00 & 0.00 \\
\multirow[c]{2}{*}{epinions} & 1.0 & \textbf{0.2293} & 0.1657 & 0.0315 & 0.0152 & \textbf{0.8294} & 0.8154 & 0.7857 & 0.7821 & 740.10 & 0.03 & 0.02 & 0.02 \\
 & 5.0 & \textbf{0.3816} & 0.2204 & 0.1446 & -0.0494 & \textbf{0.8041} & 0.7530 & 0.7290 & 0.6676 & 281.98 & 0.02 & 0.02 & 0.02 \\
\multirow[c]{2}{*}{Erd\H{o}s-R\'{e}nyi} & 1.0 & \textbf{0.2034} & 0.1435 & 0.0340 & 0.0009 & \textbf{0.8202} & 0.8063 & 0.7790 & 0.7710 & 7.71 & 0.00 & 0.00 & 0.00 \\
 & 5.0 & \textbf{0.2623} & 0.1163 & 0.1919 & 0.0034 & \textbf{0.7527} & 0.7056 & 0.7294 & 0.6697 & 5.38 & 0.00 & 0.00 & 0.00 \\
\multirow[c]{2}{*}{facebook-wosn-wall} & 1.0 & \textbf{0.2568} & 0.1719 & 0.0481 & 0.0009 & \textbf{0.8212} & 0.8007 & 0.7710 & 0.7596 & 334.43 & 0.01 & 0.01 & 0.01 \\
 & 5.0 & \textbf{0.3898} & 0.2438 & 0.1697 & 0.0087 & \textbf{0.7213} & 0.6546 & 0.6208 & 0.5473 & 393.35 & 0.01 & 0.01 & 0.01 \\
\multirow[c]{2}{*}{flickr-growth} & 1.0 & \textbf{0.1896} & 0.0999 & 0.0311 & 0.0162 & \textbf{0.8278} & 0.8087 & 0.7941 & 0.7910 & 3\,373.66 & 2.27 & 0.70 & 0.96 \\
 & 5.0 & \textbf{0.3385} & 0.2106 & 0.1413 & 0.0217 & \textbf{0.7578} & 0.7110 & 0.6856 & 0.6417 & 4\,710.21 & 2.36 & 1.03 & 1.01 \\
\multirow[c]{2}{*}{munmun} & 1.0 & \textbf{0.0982} & 0.0375 & 0.0152 & -0.0008 & \textbf{0.8224} & 0.8105 & 0.8061 & 0.8029 & 383.99 & 0.04 & 0.02 & 0.02 \\
 & 5.0 & \textbf{0.1944} & 0.0326 & 0.1058 & 0.0038 & \textbf{0.7048} & 0.6455 & 0.6724 & 0.6350 & 368.46 & 0.04 & 0.03 & 0.02 \\
\multirow[c]{2}{*}{sbm} & 1.0 & \textbf{0.3174} & 0.2385 & 0.0421 & -0.0018 & \textbf{0.8225} & 0.8020 & 0.7509 & 0.7395 & 56.68 & 0.00 & 0.00 & 0.00 \\
 & 5.0 & \textbf{0.4641} & 0.3313 & 0.1851 & -0.0070 & \textbf{0.7513} & 0.6898 & 0.6215 & 0.5326 & 107.88 & 0.00 & 0.00 & 0.00 \\
\multirow[c]{2}{*}{soc-epinions1} & 1.0 & \textbf{0.2070} & 0.1402 & 0.0265 & -0.0032 & \textbf{0.8285} & 0.8140 & 0.7894 & 0.7830 & 620.93 & 0.02 & 0.01 & 0.01 \\
 & 5.0 & \textbf{0.3504} & 0.1929 & 0.1667 & -0.0791 & \textbf{0.8033} & 0.7556 & 0.7477 & 0.6732 & 279.17 & 0.02 & 0.01 & 0.01 \\
\multirow[c]{2}{*}{vaxNoVax} & 1.0 & \textbf{0.2976} & 0.2319 & -0.0366 & -0.0042 & \textbf{0.8318} & 0.8161 & 0.7519 & 0.7596 & 94.97 & 0.06 & 0.15 & 0.10 \\
 & 5.0 & \textbf{0.2976} & 0.2319 & -0.0366 & -0.0042 & \textbf{0.8318} & 0.8161 & 0.7519 & 0.7596 & 100.95 & 0.06 & 0.12 & 0.12 \\
\bottomrule

\end{tabular}
}
\vspace{-3mm}
\end{table*}

\spara{Real-world data (R, in Table \ref{tab:datasets}).}
We use two real-world networks from Twitter, along with corresponding opinions on social issues
(Brexit and Vaccination) which were estimated with a supervised text-based classifier~\cite{minici2022cascade,cossard2020falling}.
The opinion of a user is taken as the average opinion over all their tweets.

\spara{Synthetic data (S, in Table \ref{tab:datasets}).}
We use three different types of network models:
Erd\H{o}s-R\'{e}nyi, Albert-Barab\'{a}si, and the stochastic block model.
We also consider two types of opinion distributions: uniform in the interval $[-0.5,0.5]$ and Gaussian.
To control the initial level of polarization, we incorporate one parameter that characterizes the
bi\-modality of the distributions.
First, for the uniform distribution, this parameter, denoted as $p$, is used to re\-scale each opinion value $z_u$ using the function $f(z_u)=\mathrm{sign}(z_u) |z_u|^{1/p}$.
Second, to account for structural polarization within the network, we employ
the Kernighan–Lin algorithm~\cite{kernighan1970efficient}
to partition the network into two communities and assign
to the nodes of each community opinions from two distinct Gaussian distributions.
In this setting, the polarization parameter $p$
is proportional to the distance between the means of the two communities.

\spara{Semi-synthetic data (SS, in Table \ref{tab:datasets}).}
We use public real-world networks that are semantically aligned with our study.
To generate opinions, we follow the same approach as for synthetic data.

\spara{Methods.}
We refer to our method as \ourmethod, for Lap\-la\-cian-con\-strain\-ed gradient descent.
To enhance the efficiency of \ourmethod,
we implement an acceleration technique using the ADAM optimizer~\cite{kingma2014adam}.

To evaluate \ourmethod against other methods,
we implement three realistic baselines,
which prioritize links involving specific types of nodes, namely,
\emph{neutral nodes} ($w_{u,v} \propto 1/|s_u|$),
nodes with \emph{opposite views} ($w_{u,v} \propto |s_u -s_v|$), and
\emph{popular nodes} ($w_{u,v} \propto d_u$).

\spara{Measures.}
We quantify the reduction in polarization and disagreement using two measures.
First, we compare the objective function in Eq.~\eqref{eq:objective} computed at equilibrium,
with the objective function obtained at the equilibrium with the optimized weights.
This measure, denoted as ${\rhocounterfactual}$, is defined as
\begin{equation}
{\rhocounterfactual} = 1 - \frac{P(\vec{z}^*) + D(\vec{z}^*, \mat{A}^*)}
						 {P(\vec{z}) + D(\vec{z}, \mat{A})}\,,
\label{eq:measure_counterfactual}
\end{equation}
where $\vec{z}^*=(2\eye - \mat{A}^*)^{-1}\vec{s}$ is the equilibrium produced by the optimized adjacency matrix,
and $\vec{z}=(2\eye - \mat{A})^{-1}\vec{s}$ is the equilibrium produced by the original input matrix.
This measure quantifies the reduction in polarization
while accounting for the natural reduction in disagreement achieved by the FJ model.
Additionally, we consider the reduction in polarization and disagreement
with respect to the values produced by the innate opinions and the original input matrix.
This measure, denoted as ${\rhooverall}$, is defined as
\begin{equation}
{\rhooverall} = 1 - \frac{P(\vec{z}^*) + D(\vec{z}^*, \mat{A}^*)}{P(\vec{s}) + D(\vec{s}, \mat{A})}\,.
\label{eq:measure_initial}
\end{equation}

\begin{figure*}[t]
\centering
\includegraphics[width=\textwidth]{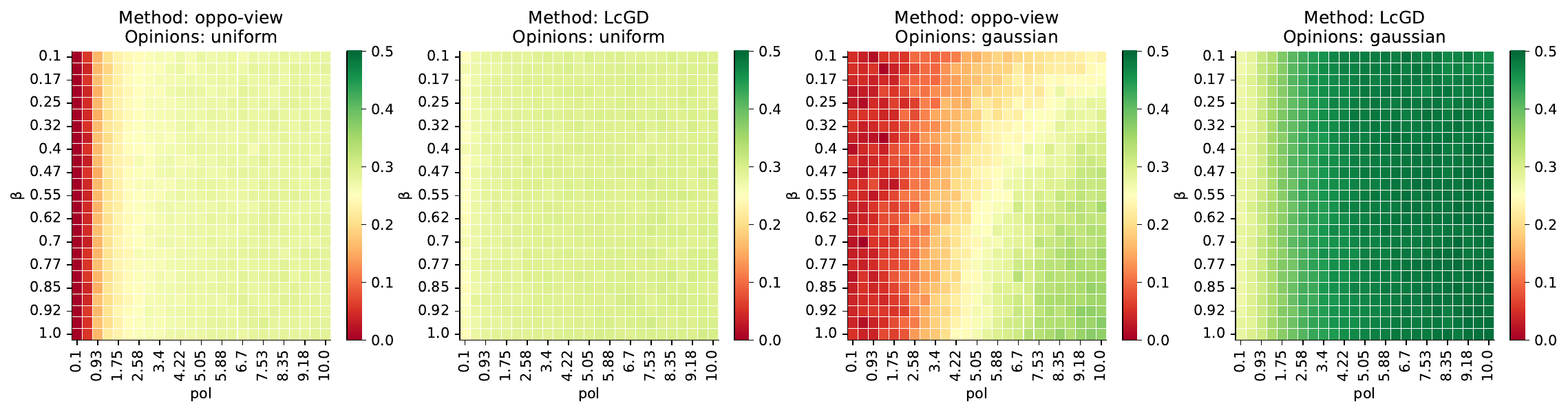}
\vspace{-6mm}
\caption{Modularity ($\beta:=$probability to connect two nodes in different communities) vs.\ polarization level in SBM network.}
\label{fig:modularity}
\vspace{-2mm}
\end{figure*}

\para{Evaluation scheme.}
Our evaluation scheme is as follows.

\para{1.} We sample equilibrium opinions $\vec{z}$ from a uniform/Gaussian distribution or,
in the case of real datasets, we use the given opinions.

\para{2.} We normalize the adjacency matrix to make it row-stochastic.

\para{3.} We use the FJ model to estimate innate opinions as in  Eq.~\eqref{eq:equilibrium2}.

\para{4.} We compute the optimal adjacency matrix $\mat{A}^*$ for each method
and the reduction in polarization and disagreement using the two measures~\rhocounterfactual and~\rhooverall.

\subsection{Results}
\para{RQ1: Performance.}
Our algorithm, \ourmethod, consistently outperforms the baselines in all input configurations.
It achieves a significant reduction in polarization and disagreement
ranging from 10\% to 50\% compared to the equilibrium with no intervention (measure \rhocounterfactual).
Compared to the initial value (measure \rhooverall), \ourmethod achieves an impressive reduction
reaching up to 80\%.
As expected, the most consistent baseline is the one that promotes neutral connections, as it increases the visibility of non-polarized content.
The baseline that strengthens connections between users with opposite views
performs the best in polarized configurations with no echo chambers.
On the other hand, promoting popular content can actually exacerbate polarization and disagreement,
resulting in poor performance across all input instances.

As shown in Fig.~\ref{fig:objective_curves} \ourmethod converges
in a small number of iterations, typically fewer than 10,
even for networks with over 2 million nodes, which corresponds to a runtime of approximately 10 minutes.
Additionally, through a grid search for the learning rate, with values 0.2, 0.02, and 0.002,
we find that a relatively large learning rate works well in terms of both time and performance.

\spara{RQ2: Effect of parameters.}
In practice, limited interventions in users' feeds can be desirable.
To address this, we conduct experiments (Fig.\ref{fig:objective_curves}) where we modify the input weights in the network by a fixed percentage.
This approach results in an adjacency matrix that is a convex combination of the projected row-stochastic gradient and the input matrix at each iteration (line 11 in Alg.~\ref{alg:algorithm}).
Remarkably, \ourmethod outperforms the baselines in configurations with polarized and segregated communities, even when changing only 75\% of the weights.

The \ourmethod algorithm demonstrates superior performance in highly-polarized configurations,
both in terms of network structure (high modularity) and opinion distribution (high polarization).
Such configurations mirror real-world social networks,
making our approach an ideal solution for such instances.
Conversely, less clustered structures,
such as Erd\H{o}s-R\'{e}nyi graphs and uniform opinion distributions,
do not exhibit a sufficiently-high level of polarization to support substantial reductions in the objective.
Despite of this, as shown in Fig. \ref{fig:modularity},
\ourmethod consistently surpasses alternative approaches
even in situations with relatively low polarization levels.
Additionally, the performance of \ourmethod remains unaffected
by the modularity of the network,
reflecting its robustness to real-world scenarios.

\begin{table}[t]
\caption{Results on \rhooverall ($\uparrow$ \textbf{better}) for undirected graphs with two distributions of opinions (uniform and gaussian).\label{table:undirected}}
\vspace{-3mm}
\let\center\empty
\let\endcenter\relax
\centering
\resizebox{\columnwidth}{!}{\begin{tabular}{lrrrrr}
\toprule
 &  & \multicolumn{2}{c}{gaussian} & \multicolumn{2}{c}{uniform} \\
\cmidrule(lr){3-4} \cmidrule(lr){5-6} 
Network & nodes & \ourmethod & SCS & \ourmethod & SCS \\
\midrule
Albert-Barab\'{a}si graph & 100 & \textbf{0.8154} & 0.7990 & \textbf{0.8568} & 0.8560 \\
Erd\H{o}s-R\'{e}nyi graph & 100 & \textbf{0.7677} & 0.7206 & \textbf{0.8584} & 0.8558 \\
Aarhus CS department & 61 & \textbf{0.7650} & 0.7481 & 0.8249 & \textbf{0.8272} \\
dimacs10-football & 115 & \textbf{0.7588} & 0.7293 & 0.8322 & \textbf{0.8358} \\
political books & 105 & \textbf{0.7046} & 0.6101 & \textbf{0.8439} & 0.8369 \\
students cooperation & 141 & \textbf{0.5324} & 0.5269 & \textbf{0.8462} & 0.8335 \\
\bottomrule
\end{tabular}
}

\vspace{3mm}

\caption{Gradient descent scheme comparison on \rhooverall ($\uparrow$ \textbf{better}) for directed graphs with gaussian opinions.}
\let\center\empty
\let\endcenter\relax
\centering
\vspace{-4mm}
\resizebox{.9\columnwidth}{!}{\begin{tabular}{lrrrr}
\toprule
 & \multicolumn{2}{c}{\rhooverall} & \multicolumn{2}{c}{Time (sec)} \\
 \cmidrule(lr){2-3} \cmidrule(lr){4-5} 
Network & \ourmethod & simple GD & \ourmethod & simple GD \\
\midrule
brexit & \textbf{0.8216} & 0.8215 & \textbf{38.09} & 79.23 \\
digg-friends & 0.7830 & \textbf{0.7873} & \textbf{354.45} & 505.09 \\
ego-twitter & \textbf{0.6671} & 0.6655 & \textbf{79.31} & 94.21 \\
epinions & 0.8041 & \textbf{0.8071} & \textbf{281.98} & 969.03 \\
vaxNoVax & \textbf{0.8318} & 0.8299 & \textbf{100.95} & 193.51 \\
\bottomrule
\end{tabular}
}
\label{table:ablation}
\vspace{-3mm}
\end{table}

\spara{RQ3: Undirected graphs.}
We also compare \ourmethod
with the SDP approach (defined in Eq.~\eqref{eq:sdp}),
using both synthetic and real-world undirected networks.
We solve the SDP problem using the CVX library~\cite{grant2014cvx} with the SCS solver~\cite{o2016conic}.
Due to the computational limitations associated with the SDP approach  \cite{majumdar2019survey},
we limit the experiments to small networks.
As discussed in Section 2, the SDP formulation can be viewed as a version of the method proposed by Musco et al.~\cite{musco2018minimizing}, adapted to accommodate a given graph topology.
It is worth stressing that in this experiment, both the methods take in input an undirected graph, but while the SDP is forced to output a weighted undirected graph, our more general approach is allowed to assign different weights to $(i,j)$ and $(j,i)$.

As shown in Table~\ref{table:undirected},
\ourmethod outperforms the SDP approach,
especially in the case of Gaussian opinion distributions.
For uniform opinion distributions,
where there is a lack of structural polarization patterns,
the two methods are on par,
and in two configurations the SDP approach exhibits slightly better performance.

In summary, our method for directed graphs is also applicable to undirected graphs, where, by treating an undirected edge as the two corresponding directed links, it is possible to assign different importance weights to the content produced by user $i$ for the feed of user $j$ with respect to the content produced by $j$ for the feed of $i$. This flexibility, makes it to overperform the SDP method, which is optimal, but only for symmetric solutions.

\spara{Ablation study.}
We assess the impact of the ADAM optimizer
through an ablation study where we use only the standard gradient-descent procedure.
The results are presented in Table \ref{table:ablation}.
We observe that both algorithms achieve comparable values of the objective,
but \ourmethod, which uses the ADAM optimizer, converges faster than standard gradient descent. 

\section{Conclusions}
\label{sec:conclusions}
Leveraging the Friedkin--Johnsen opinion-dynamics model, we
presented a novel approach to minimize polarization and disagreement in a social network.
This is done by re-weighting the relative importance of the accounts that a user follows,
so as to calibrate the frequency with which the content produced by various accounts is shown to the user. By re-balancing the weight of existing links only, our approach preserves the total engagement of each user in the social network.
We showed that our objective function is
not matrix-convex, but the feasible set is convex, and we devised a scalable algorithm based on projected gradient descent. With an extensive experimental evaluation, we demonstrated
the effectiveness of the proposed approach in its ability to minimize polarization and disagreement in large social networks.

While the paper contributes valuable insights, it is important to acknowledge some limitations.
The absence of \emph{backfire-effect} modeling~\cite{wood2019elusive}
is inherent in the FJ model and restricts the approach's ability
to fully capture the complexities of polarization dynamics.
Further research could explore incorporating this effect.
Furthermore, it would be interesting to extend our approach to handle multiple topics and
leverage additional features.

\section*{ACKNOWLEDGMENTS}
This research is supported by ERC Advanced Grant REBOUND (834862), the EC H2020 RIA project SoBigData++ (871042), and the Wallenberg AI, Autonomous Systems and Software Program (WASP) funded by the Knut and Alice Wallenberg Foundation.

\bibliographystyle{ACM-Reference-Format}
\balance
\bibliography{reference}

\end{document}